\documentclass[review,12pt]{elsarticle}

\usepackage{amssymb}
\usepackage[a4paper,lmargin=3cm,rmargin=2cm,vmargin=2cm]{geometry}
\usepackage[scaled=0.84]{helvet}
\usepackage{makeidx}
\usepackage{parskip}
\usepackage[ansinew]{inputenc}
\usepackage{mathrsfs}
\usepackage{amsmath,amsthm,amsfonts,amssymb}
\usepackage{graphicx}
\usepackage[english]{babel}
\usepackage{hyperref}
\usepackage{enumerate,fancyhdr}
\usepackage{etoolbox,lineno}
\usepackage{fouridx}
\usepackage{epstopdf}
\usepackage{epsfig}
\usepackage{xcolor}
\numberwithin{equation}{section}
\theoremstyle{plain}
\newtheorem{theorem}{Theorem}[section]
\newtheorem{proposition}[theorem]{Proposition}
\newtheorem{lemma}[theorem]{Lemma}
\newtheorem{corollary}[theorem]{Corollary}

\journal{Insurance Mathematics \& Economics}

\begin{document}
	
	\begin{frontmatter}
		
		
		
		\title{State Space Vasicek Model of a Longevity Bond}
		
		
		\author{KALU,Georgina Onuma}
		\address{Department of Mathematical Sciences University of South Africa, South Africa}
		\author{IKPE,Chinemerem Dennis}
		\address{Michighan State University and African Institute for Mathematical Sciences, South Africa }
		\author{ORUH Benjamin Ifeanyichukwu}
		\address{Department of Mathematics Michael Okpara University of Agriculture,Umudike, Nigeria}
		\author{GYAMERAH, Samuel Asante}
		\address{Department Stattistics and Actuarial Science, Kwame Nkrumah University of Science and Technology, Ghana}

		\begin{abstract}
			Life expectancy have been increasing over the past years due to better health care, feeding and conducive environment. To manage future uncertainty related to life expectancy, various insurance institutions have resolved to come up with financial instruments that are indexed-linked to the longevity of the population. These new instrument is known as longevity bonds . In this article, we present a novel classical Vasicek one factor affine model in modelling zero coupon longevity bond price (ZCLBP) with financial and mortality risk. The interest rate $r(t)$ and the stochastic mortality $\lambda(t)$ of the constructed model are dependent but with uncorrelated driving noises. The model is presented in a linear state-space representation of the contiuous-time infinite horizon and used Kalman filter for its estimation.  The appropriate state equation and measurement equation derived from our model is used as a method of pricing a longevity bond in a financial market.
			 The empirical analysis results show that the unobserved instantaneous interest rate shows a mean reverting behaviour in the U.S. term structure. The  zero-coupon bonds yields are used as inputs for the estimation process. The results of the analysis are gotten from the monthly observations of U.S. Treasury zero coupon bonds from December, 1992 to January, 1993. The empirical evidence indicates that to model properly the historical mortality trends at different ages, both the survival rate and the yield data are needed to achieve a satisfactory empirical fit over the zero coupon longevity bond term structure. The dynamics of the resulting model allowed us to perform simulation on the survival rates, which enables us to model longevity risk. 
			\end{abstract}
		
		\begin{keyword}
			Zero coupon Longevity bond, Kalman filter, Survival rate, Term structure, Stochastic mortality, Interest rate.
		\end{keyword}
		
	\end{frontmatter}
	\section{Introduction}
	Human life expectancy has been increasing significantly for the last few decades, especially in the developed countries (example United Kingdom) where changes in lifestyle, medical advancement, genetic new discoveries , etc have been playing a good role. These life expectancy have proved to have an important effect at higher ages and have caused life offices (for example Pension funds, annuity providers, the government and Defined Benefits (DB) pension schemes) to incur losses on their annuity life business. The insurance industry is therefore taking the responsibility of the cost of the unanticipated higher longevity by issuing longevity bonds.
		
	Longevity bonds can be defined as  annuity bonds whose coupons are not fixed over time but fall in line with a given survivor index. Longevity bonds were first proposed by \cite{Blake and Burrows (2001)}, with its first operational mortality-linked in 2003. In November 2004, BNP Paribas (in its role as structurer and manager) announced that the European Investment Bank (EIB) would issue a longevity bond with a maturity of twenty-five years. According to UK Office for National Statistics (ONS), the coupon payments were to be linked to a survivor index that is based on the realized mortality experience of a population of males from England and Wales aged 65 in 2003. Although the UK life offices and pension funds were the main intended investors, but then, this issue was unsuccessful at last. Different authors, for example \cite{Biffis (2005)} have modeled the mortality intensity as a stochastic process which allowed them to capture two important features of the mortality intensity: time dependency and uncertaintity of the future development. \cite{Schrager (2006)}  proposed a new model for stochastic mortality. The model was based on affine term structure model that satisfies three important requirements for application in practice: analytical tractibility, clear interpretation of the factors and compatibility with financial option pricing models. They tested the model fit using data on Dutch mortality rates and it shows that despite its shortcomings, in general it fits the data well and it is easily applicable to the pricing of well known embedded option. \cite{Dehl M. (2004)} focus on mortality at pensionable ages and showed how the risk of longevity can be taken into account. However, longevity improvement is in an unpredicted way since longevity bond serves as an instrument to hedge against longer longevity risks faced by life offices and pension plans. 
	
	In recent years, different authors have studied  longevity bond pricing. \cite{Johnny Siu-Hang Li (2010)} discussed another approach of longevity pricing model using the parametric boot strap. \cite{Chena et al (2010)} used three popular methods - first neutral method, sharp ratio rule and the wang transform and \cite{Kogurea et al (2010)} proposed a Bayesian approach.
	\cite{Thomas et al (2013)} proposed  a simple Partial internal model for longevity risk inside the framework of Solvency II. Their model tend to have connection with the Danish longevity benchmark mechanisms and it comprises of a component, which is based on the size of a given portfolio. This component can measure for each insurance portfolio the unsystematic longevity risk .
	In other to tackle the challange of measuring and modeling the longevity risk, many studies made use of the affine structure with mortality models application. For example \cite{Dehl M. (2004)} suggested for a cohort live an
	affine mortality structure with the same ages \cite{Schrager (2006)} introduced an affine mortality intensity model specially for Thiele and Makeham mortality laws and examines all ages in same time.
	Furthermore, \cite{Tzuling and Tzeng (2010)} gave a variable method to modelling longevity risk while \cite{Andrew J.G. Cairns (2011)} proposed a mortality model with an age shift by using a principal component analysis PCA. \cite{Sharon S. et al (2010)} used the Principal Component Approach to model longevity risk and compared his approach with the existing stochastic mortality models. They went further using their model to examine the ratio of annuity price both for the deferred
	life annuity and life annuity products and then try to forecast the future mortality rates. They were able to prove that their model can resolve the Lee Carter Model’s problem by using the PCA approach.\cite{Ronkainen (2012)} used the famous \cite{Lee and Carter (1992)} approach to develop a stochastic modeling of financing longevity risk in pension insurance and then used the Bayesian MCMC methods and prove that the LC model is not completely adaptable with mortality data.

	There are two sources of mortality risk that affect a portfolio of pensions or annuities: the unsystematic (idiosyncratic) and systematic (longevity) risk. In the case of the systematic mortality-related, there are different  possible numbers of risks faced by annuity providers and life insurers. In this paper, the term mortality risk will be taken to include uncertainty of all forms in future mortality rates,  increases and decreases in the rate mortality inclusive. Most of the times, longevity risk is normally taken to mean the risk that survival rates are higher than anticipated. But in this paper, is taken to mean uncertainty in either direction. In academic literature, several supporting instruments such as survivor bonds or longevity (cf.\cite{Blake and Burrows (2001)}) or survivor swaps (cf. \cite{Dowd et al (2006)}) have been proposed (see \cite{Blake et al (2006)} for an overview). 
	
	Different approaches for the modelling of aggregate  randomness in mortality rates over time have been proposed . A key earlier work is that of \cite{Lee and Carter (1992)}. In their work, they focused on the practical application of stochastic mortality and its statistical analysis. Because aggregate mortality rates are better of measured annually, \cite{Lee and Carter (1992)} and other authors (for example\cite{Brouhns et al (2002)}; \cite{Renshew et al (2003)}; \cite{Currie et al (2004)}) have adopted a similar approach and worked in discrete time model. Models following the approach of \cite{Lee and Carter (1992)} typically adapt discrete-time series models. This is to capture the random element in the stochastic development of mortality rates. Similarly, different authors (see, for example, \cite{Milevsky and Promislow (2001)}; \cite{Dehl M. (2004)}; \cite{Dahl and Moller (2005)}; \cite{Miltersen and Persson (2005)}; \cite{Biffis (2005)};\cite{Schrager (2006)}) have developed models in a continuous-time framework. This is because discrete model parameters can be calibrated to historical mortality experience while continuous-time model has more benefit in areas such as valuation and hedging of life insurance liabilities especially when we combine the mortality model with some financial analysis. 
	Hence, continuous-time models have a key role to play when it comes to understanding how the prices of mortality-linked securities develops over time. 
	
	In this article, we addressed how to construct a market- consistent-valuation framework for long-term risky investment in the insurance and financial market. And also how interest rate models can be constructed to obtain analytically tractable and accurate pricing formulae for long-term assets. To achieve this we formulate a linear state-space representation of the continuous- time infinite horizon short rate models and use Kalman filter to jointly estimate the current term structure and its dynamics for markets with illiquid long-term bond. Analytical tractability is the most important of what we intend to achieve.This is because we want to take into consideration the parameter uncertainty and show how our result can be used to price new zero coupon longevity bond.  Consequently, we choose to develop a model in continuous-time and adopt an approach that is similar to that of \cite{Vasicek (1977)}.
	We propose a zero coupon longevity bond price model that is used to fit U.S. Treasury bond maturity data and show how the calibrated model can be used to price financial and mortality risk in longevity bond. To estimate our model , we use a state space formulation which contains the measurement and the transition equations. The model involves two stochastic factors - the interest rate and mortality risk. While former takes into account the yield data, the later takes into account the real survival rate. We present empirical evidence that indicates that to model sufficiently historical mortality longevity at different ages, these factors (the interest rate and mortality risk) are necessary in other to achieve a reasonable empirical fit over the zero coupon longevity bond term structure. Unlike the
	Econometrics models as that of the famous \cite{Lee and Carter (1992)}, our transition equations which
	are the time-series dynamics are presented in the model like a stochastic factor which allow us to simulate survival rates, thereby enabling us to model longevity risk. The measurement equation represents the yields observed data which is exponentially affine in the factors.
		
	For pricing the zero coupon longevity bond, we adopt the risk-adjusted (or “risk-neutral”)
	approach of pricing adopted by, for example,\cite{[Musiela and Rutkowski(1998)}. We suggest a simple method for making the adjustment between risk-adjusted and real probabilities, which involves a constant market price both for the longevity and parameter risk. This is because of the current lack of market data. 
	
	The remaining of this paper is organized as follows. In section 2, we discuss the relationship between financial and motality longevity; present the formular for the mortality rate and interest rate equation; discuss the meaning of longevity bond as a financial product; present asymptotic yield for the infinite horizon and give the zero coupon longevity bond price for the constant parameter case.The state space formulation of the model and parameters estimation is discuss in section 3. In section 4, the emperical and data analysis were presented. In section 5, some conclusions are offered.
	
	\section{Financial and Mortality Longevity Rate Model}
	Longevity bonds are the first financial products to offer longevity protection by hedging the trend in longevity.
	Longevity bonds are needed because life expectancy has been constantly increasing as a result of  medical improvements, better life standards, etc. As pointed out by \cite{Hardy} that life expectancy for men aged 60 is more than 5 years' longer in 2005 than anticipated to be in mortality projections made in the 1980 and as such, this has brought about uncertainty of longevity projections.
	
	In other to meet up with this demand, the Capital markets offer longevity bonds with coupons depending on the survival rate of a given population which can be used to hedge a big portion of the longevity risk. Longevity bonds
	can take a large variety of forms which can vary enormously in their sensitivities to longevity shocks. Hence, longevity bonds is ideal asset for hedging the longevity risk of a pension fund. 
	\subsection{Zero Coupn Longevity Bond: The term structure equation}
	Given a probabilty space $(\Omega,\mathcal{F},\mathcal{F}_{t(t\geq 0)},\mathcal{P})$, which satisfies the usual hypothesis that the filtration $\mathcal{F}_{t(0\leq t \leq T)}$ is right continuous with left limits where $\mathcal{P}$ is the joint process of interest rate, then the mortality rate is given by 
	\begin{equation}\label{main}
	z(t)=y_0+y_1X_1(t) 
	\end{equation}
	for $z, y_0, y_1 $ and $X_1$ in $\mathcal{R}^2$\\
	Let $y_1=\left(\begin{array}{c}
	-1\\
	-1
	\end{array}\right)$ and $y_0 =\mu = \left(\begin{array}{c}
	\mu^r\\
	\mu^\lambda
	\end{array}\right)$ , equation \ref{main} becomes
	\begin{equation}\label{long}
	z(t)=\mu-X_1(t)
	\end{equation}
	\begin{align}
	z(t) =  \left( \begin{array}{c}
	r(t) \\
	\lambda(t)
	\end{array} \right) =\left( \begin{array}{c}
	\mu^r \\ 
	\mu^\lambda\\     
	\end{array} \right)-\left( \begin{array}{c}
	X_1^r(t)\\
	X_1^\lambda(t)
	\end{array}\right)  
	\end{align}
	Let the dynamics of ${X_1}$ be given as
	\begin{equation}
	dX_1=-\zeta_1 X_1dt + c_1dW_1,\label{dyna}
	\end{equation}
	where $\zeta_1$ is the drift, $c_1$ is the diffusion, and $dW_1$ is two dimensional vector. From equation \ref{dyna} 
	\begin{equation}
	dX_1=-\zeta_1 X_1dt + \kappa_{11}dZ_1
	\end{equation}
	where $dZ_1$ is two dimensional vector and are independent and $\kappa_{11}$ is scaler.
	The dynamics of  the process $z$ is given by the following stochastic differential equation. 
	\begin{equation}
	dz(t)=\zeta_1(\mu-z(t))dt-c_1dW_1(t)
	\end{equation}
	$\zeta_1$ is the mean reversion rate while $c_1$ is the diffusion term. We also assume that a market exist for bonds of every available maturity and that the market is abitrage free. 
		\subsection{Affine Form of Longevity Bond Price}
	There are several ways of solving a longevity problem and one of it is classical Affine term structure model. Generally, affine term structure models are useful in modelling mortality intensity in the literature. \cite{Luciano et al. (2005)} described the mortality intensity by affine models and calibrated the intensity processes by using observed and projected UK mortality
	tables. They claimed that affine processes with deterministic part increases exponentially which could describe the evolution of mortality intensity properly. \cite{Russo et al. (2011)} calibrated three affine stochastic mortality models using term assurance premiums of three Italian insurance companies and proposed that those affine models can be used for pricing mortality-linked securities. 
	
	The theory of Affine model rank among the most popular model both in theory and practice, and as such many examples have been investigated. The first term structure model-Vasicek model is an affine model. Others are Cox, Ingersol, Ross (CIR), Hull and White or Longstaff and Schwartz which are also Affine model.\\
	In the bid of the expansion of the Affine model,\cite{Duffie and Kan (1996)}  investigated the model and developed a general theory while \cite{Dai and Singleton (2000)} provided the classification and thereafter established most of the general representative of each class of Affine model.\\The tractability and flexibilty of Affine model makes its popularity to be known and aslo there are often explicit solution for the bond price and bond option price. In this work, we will  be considering another form of classical Vasicek one-factor model of a longevity bond price whose \textit{zero coupon longevity bond} (ZCLB) is given as:
	\begin{theorem} \label{bond equation} \textit{(Babbs, S. H., and M. J. P. Selby 1993)} The price of a zero coupon longevity bond with maturity at time $S$ is given as
		\begin{multline}\label{bond}
		B_l(S,t)=\exp\left\lbrace-\int_{t}^{S}\left\vert\mu(u)\right\vert du-\int_{t}^{S}\theta_1\left\vert\sigma_1\right\vert(S,u)-\frac{1}{2}\left\vert\sigma_1\right\vert^2(S,u)du+\right.\\ \left. \frac{G_1(S)- G_1(t)}{G_1^\prime(t)} X_1(t)\right\rbrace
		\end{multline}
		where 
		$\left\vert.\right\vert =L_1$ norm, $0\leq t\leq S$
		\begin{align}
		G_1&=\int_{0}^{t}\exp\left\lbrace-\int_{0}^{u}\zeta_1(s)ds\right\rbrace du, \quad and\\&=\frac{1}{\zeta_1}(1-e^{-\zeta_1(S-u)})\\
		\sigma_1(S,t)&= \frac{G_1(S)-G_1(t)}{G_1^\prime(t)}\kappa_{11}(t)\\
		&=\tau H(\zeta_1\tau)X_1(t)
		\end{align}
	\end{theorem}
	\begin{proof}[Proof of theorem]
		The detailed proof of theorem \ref{bond equation} can be found in \cite{Babbs and Nowman (1993)}
	\end{proof}
The price at time $t$ of a zero coupon longevity bond $B_l(t,S)$ under the vasicek factor model with maturity $S$ can be written as:
\begin{equation}
B_l(t,S,X_y(t))=\mathbb{E}_z^\mathbb{Q}\left[\exp\left(-\int_{t}^{S}\left\vert z\right\vert(u,X_y(u)) du\right)\right],\ \ \ B_l(S,S)=1.
\end{equation}
where $\mathbb{Q}$ is a martingale measure. 	\begin{corollary}
		The continuously compounded zero coupon longevity yield $R_l(S,t)$ is given by
		\begin{equation}\label{constant}
		R_l(S,t)=\frac{1}{\tau}\left(\int_{t}^{S}\left\vert\mu(u)\right\vert du+\int_{t}^{S}\theta_1\left\vert\sigma_1\right\vert(S,u)-\frac{1}{2}\left\vert\sigma_1\right\vert^2(S,u)du-\frac{G_1(S)- G_1(t)}{G_1^\prime(t)}X_1(t)\right)
		\end{equation}
		where 
		\begin{equation*}
		\int_{t}^{S}\left\vert\mu\right\vert du=\mu\tau,\ \ \ \ \tau=S-t
		\end{equation*}
	\end{corollary}
	\begin{proof}
		This is a direct consequence of theorem \ref{bond equation}
	\end{proof}
	\subsection{Asymptotic yield for Infinite horizon}
	Given the maturity of a zero coupon longevity bond is infinite and all the parameters remain constant, then the folowing proposition holds true.
	\begin{proposition}
		Let $R_l(\infty)$ represent the yield on a zero coupon longevity bond with infinite maturity. Since $\tau=S-t$ for $t\leq S$ then $S\rightarrow\infty$ implies $\tau\rightarrow\infty$. Consequently
		\begin{equation}
		R_l(\infty)=\lim\limits_{S\rightarrow\infty}R_l(S,t)=\lim\limits_{\tau\rightarrow\infty}R_l(\tau+t,t)
		\end{equation}
		On constant parameters  equation\ref{constant} becomes
		\begin{equation}\label{const para}
		R_l(S,t)=\left[R_l(\infty)-w(\tau)-H(\zeta_1\tau)X_1\right],\ \ \ \  t\in[0,S]
		\end{equation}
		with
		\begin{subequations}
			\begin{align}
			R_l(\infty)&=\mu+\theta_1\frac{\kappa_{11}}{\zeta_1}-\frac{1}{2}\left(\frac{\kappa_{11}}{\zeta_1}\right)^2,\\
			w(\tau)&= H(\zeta_1\tau)\left[\theta_1\frac{\kappa_{11}}{\zeta_1}-\frac{\kappa_{11}\kappa_{12}}{{\zeta_1\zeta_1}}\right]+\frac{1}{2}((H{\zeta_1+\zeta_2})\tau)\frac{\kappa_{11}\kappa_{12}}{{\zeta_1\zeta_2}},
			\end{align}
		\end{subequations}
	\end{proposition}
	\begin{proof}
		From  equation\ref{bond}  we have
		\begin{equation*}
		\frac{G_1(S)- G_1(t)}{G_1^\prime(t)} X_1(t)=\tau\frac{(1-e^{-\zeta_1\tau})}{\zeta_1\tau}X_1(t)
		\end{equation*}
		\begin{equation}
		=\tau H(\zeta_1\tau)X_1(t)
		\end{equation}
		where $H(\zeta_1\tau)=
		\left(\begin{array}{cc}
		\frac{(1-e^{-\zeta_1^r\tau})}{\zeta_1^r\tau}&0\\
		0&\frac{(1-e^{-\zeta_1^\lambda\tau})}{\zeta_1^\lambda\tau}
		\end{array}\right)$.\\ Also
		\begin{align*}
		\sigma_1(S,u)&= \frac{G_1(S)- G_1(u)}{G_1^\prime(u)}\kappa_{11},\\&=\frac{1}{\zeta_1}(1-e^{-\zeta_1(S-u)})\kappa_{11}\\ 
		\end{align*}
		From  equation \ref{constant} let the two $\int_{t}^{S}$  
		\begin{equation*}
		E(\tau)=\int_{t}^{S}\left(\theta_1\left\vert\sigma_1\right\vert(S,u)-\frac{1}{2}\left\vert\sigma_1\right\vert^2(S,u)\right)du
		\end{equation*}
		be represented by $\Psi$ and $\Sigma$
		\begin{equation}
		\begin{aligned}
		\int_{t}^{S}\Psi \ \ \ &= \int_{t}^{S}\theta_1\sigma_1(S,u)du \\&=\theta_1\int_{t}^{S}\left(\frac{1}{\zeta_1}(1-e^{-\zeta_1(S-u)})\right)\kappa_{11}du\\&=\frac{\theta_1}{\zeta_1}\int_{t}^{S}(1-e^{-\zeta_1(S-u)})\kappa_{11}du\\&=\frac{\theta_1\kappa_{11}}{\zeta_1}\left[(S-t)-\frac{1}{\zeta_1}(1-e^{-\zeta_1(S-t)})\right]du\\&=\frac{\theta_1\kappa_{11}}{\zeta_1}\left(\tau-H(\zeta_1\tau)\right)\\&=\tau \frac{\theta_1\kappa_{11}}{\zeta_1}\left((1-H(\zeta_1\tau)\right)
		\end{aligned}
		\end{equation}
		\begin{equation}
		\tau \left(\begin{array}{c}
		\frac{\theta_1\kappa_{11}}{\zeta_1^r}\\\frac{\theta_1\kappa_{11}}{\zeta_1^\lambda}
		\end{array}\right)\left(\begin{array}{cc}
		1-H(\zeta_1^r\tau) &0\\
		0&1-H(\zeta_1^\lambda\tau)
		\end{array}\right)
		\end{equation}
		\begin{equation}
		\begin{aligned}
		\int_{t}^{S}\Sigma &= \int_{t}^{S}(\sigma_1)^2(S,u)du,\\&= \int_{t}^{S}\left(\frac{1}{\zeta_1}(1-e^{-\zeta_1(S-u)})\kappa_{11}\right)^2du,\\&=\int_{t}^{S}\frac{1}{\zeta_1}(1-e^{-\zeta_1(S-u)})\kappa_{11}\frac{1}{\zeta_2}(1-e^{-\zeta_2(S-u)})\kappa_{11}du\\&=\left(\frac{\kappa_{11}}{\zeta_1}\right)^2\int_{t}^{S}\left(1-e^{-\zeta_1(S-u)}-e^{-\zeta_2(S-u)}+e^{-(\zeta_1+\zeta_2)(S-u)}\right)du,
		\end{aligned}
		\end{equation}
		\begin{multline}
		=\left(\frac{\kappa_{11}}{\zeta_1}\right)^2\left((S-t)-\frac{1}{\zeta_1}(1-e^{-\zeta_1(S-t)})-\frac{1}{\zeta_2}(1-e^{-\zeta_2(S-t)})\right. \\ \left.+\frac{1}{\zeta_1+\zeta_2}1-e^{-(\zeta_1+\zeta_2)(S-t)}\right)
		\end{multline}
		\begin{align}
		&=\left(\frac{\kappa_{11}}{\zeta_1}\right)^2\left(\tau-H(\zeta_1\tau)-H(\zeta_2\tau)+\frac{1}{\zeta_1+\zeta_2}((H{\zeta_1+\zeta_2})\tau)\right)\\&=\tau\left(\frac{\kappa_{11}}{\zeta_1}\right)^2+\frac{\kappa_{11}\kappa_{12}}{{\zeta_1\zeta_1}}H(({\zeta_1+\zeta_2})\tau)-2\tau\frac{\kappa_{11}\kappa_{12}}{{\zeta_1\zeta_2}}H(\zeta_{12}\tau)
		\end{align}
		\begin{align}
		\therefore E(\tau)&=\tau\frac{\theta_1\kappa_{11}}{\zeta_1}(1-H(\zeta_1\tau))-\frac{1}{2}\left[\tau\left(\frac{\kappa_{11}}{\zeta_1}\right)^2+\frac{\kappa_{11}\kappa_{12}}{{\zeta_1\zeta_2}}((H{\zeta_1+\zeta_2})\tau)-2\tau\frac{\kappa_{11}\kappa_{12}}{{\zeta_1\zeta_2}}H(\zeta_{12}\tau)\right]\\&=\tau\frac{\theta_1\kappa_{11}}{\zeta_1}(1-H(\zeta_1\tau))-\frac{1}{2}\tau\left(\frac{\kappa_{11}}{\zeta_1}\right)^2-\frac{1}{2}\frac{\kappa_{11}\kappa_{12}}{{\zeta_1\zeta_2}}((H{\zeta_1+\zeta_2})\tau)+\tau\frac{\kappa_{11}\kappa_{12}}{{\zeta_1\zeta_2}}H(\zeta_{12}\tau)\\&=\tau\left[\frac{\theta_1\kappa_{11}}{\zeta_1}(1-H(\zeta_1\tau))-\frac{1}{2}\left(\frac{\kappa_{11}}{\zeta_1}\right)^2\right]+\tau\left[-\frac{1}{2}\frac{\kappa_{11}\kappa_{12}}{{\zeta_1\zeta_2}}((H{\zeta_1+\zeta_2})\tau)+\frac{\kappa_{11}\kappa_{12}}{{\zeta_1\zeta_2}}H(\zeta_{12}\tau)\right]\label{equaton final}
		\end{align}
		Substituting  equation \ref{equaton final} into equation \ref{constant} and taking limit as $\tau \rightarrow\infty$ equation \ref{constant} becomes
		\begin{multline}
		R_l(S,u)=\frac{1}{\tau}\left(\int_{t}^{S}\left\vert\mu(u)\right\vert du+\tau\left[\frac{\theta_1\kappa_{11}}{\zeta_1}(1-H(\zeta_1\tau))-\frac{1}{2}\left(\frac{\kappa_{11}}{\zeta_1}\right)^2\right]\right.\\ \left.+\tau\left[-\frac{1}{2}\frac{\kappa_{11}\kappa_{12}}{{\zeta_1\zeta_2}}((H{\zeta_1+\zeta_2})\tau)+\frac{\kappa_{11}\kappa_{12}}{{\zeta_1\zeta_2}}H(\zeta_{12}\tau)\right]+\tau H(\zeta_1\tau)X_1(t)\right)
		\end{multline}
		\begin{multline}
		=\mu+\frac{\theta_1\kappa_{11}}{\zeta_1}-\frac{\theta_1\kappa_{11}}{\zeta_1}H(\zeta_1\tau)-\frac{1}{2}\left(\frac{\kappa_{11}}{\zeta_1}\right)^2-\frac{1}{2}\tau \frac{\kappa_{11}\kappa_{12}}{{\zeta_1\zeta_2}}((H{\zeta_1+\zeta_2})\tau) \\ +\frac{\kappa_{11}\kappa_{12}}{{\zeta_1\zeta_1}}H(\zeta_{12}\tau)+\tau H(\zeta_1\tau)X_1(t)
		\end{multline}
		As $\tau\rightarrow\infty$ 
		\begin{align*}
		\lim\limits_{\tau\rightarrow\infty}H(\zeta_1)&=0\\
		\lim\limits_{\tau\rightarrow\infty}H(\zeta_{12})&=0\\
		\lim\limits_{\tau\rightarrow\infty}((H{\zeta_1+\zeta_2})\tau)&=0
		\end{align*}
		\begin{equation}
		\therefore R_l(\infty)=\mu+\frac{\theta_1\kappa_{11}}{\zeta_1}-\frac{1}{2}\left(\frac{\kappa_{11}}{\zeta_1}\right)^2
		\end{equation}
		\begin{equation}
		=\left(\begin{array}{c}
		\mu^r\\ \mu^\lambda
		\end{array}\right)+
		\theta_1
		\left(\begin{array}{cc}
		\frac{\kappa_{11}}{\zeta_1^r}&0\\0& \frac{\kappa_{11}}{\zeta_1^\lambda}\end{array}\right)
		-\frac{1}{2}\left(\begin{array}{cc}
		\frac{\kappa_{11}}{\zeta_1^r}&0\\0& \frac{\lambda\kappa_{11}}{\zeta_1^\lambda}
		\end{array}\right)^2
		\end{equation}
	\end{proof}
	where 
	\begin{align*}
	\tau\equiv S-t \ \ \ \ \ \ 
	\qquad {and}\\
	H(\gamma)=\frac{1-e^{-\gamma}}{\gamma}
	\end{align*}
	$H$ is a function of $\zeta_1$, $\zeta_2$ and $\tau;\gamma$ is a dummy variable.
	\subsubsection{On Constant Parameter Zero Coupon Longevity Bond Price}
	Under constant parameters assumption of the vasicek longevity model \ref{bond}, the price of a zero coupon longevity bond of any maturity of  equation \ref{bond} is given as in Lemma \ref{lem} 
	\begin{lemma} \label{lem}
		In the case where all the parameters (the long run average rate $\mu$, the market price of risk process $\theta_1$ the mean reversion $\zeta_1$ and diffusion $\kappa_{11}$ coefficents) are all constant, the longevity bond price can be defined as
		\begin{equation}
		B_l(S,t)=\exp\left\lbrace-\tau\left[R_l(\infty)-w(\tau)-H(\zeta_1\tau)X_1\right]\right\rbrace,\ \ \ \  t\in[0,S]
		\end{equation}
		with
		\begin{subequations}
			\begin{align}
			R_l(\infty)&=\mu+\theta_1\frac{\kappa_{11}}{\zeta_1}-\frac{1}{2}\left(\frac{\kappa_{11}}{\zeta_1}\right)^2,\label{constA}\\
			w(\tau)&= H(\zeta_1\tau)\left[\theta_1\frac{\kappa_{11}}{\zeta_1}-\frac{\kappa_{11}\kappa_{12}}{{\zeta_1\zeta_2}}\right]+\frac{1}{2}((H{\zeta_1+\zeta_2})\tau)\frac{\kappa_{11}\kappa_{12}}{{\zeta_1\zeta_2}},\label{constB}
			\end{align}
		\end{subequations}
	\end{lemma}
	\begin{proof}
		Substituting  equation \ref{constA} and  equation \ref{constB} into equation \ref{bond} 
		\begin{multline}
		B_l(S,t)=\exp\left\lbrace-\int_{t}^{S}\left\vert\mu(u)\right\vert du-\tau\left[\frac{\theta_1\kappa_{11}}{\zeta_1}(1-H(\zeta_1\tau))-\frac{1}{2}\left(\frac{\kappa_{11}}{\zeta_1}\right)^2\right]\right.\\ \left.+\tau\left[-\frac{1}{2}\frac{\kappa_{11}\kappa_{12}}{{\zeta_1\zeta_2}}((H{\zeta_1+\zeta_2})\tau)+\frac{\kappa_{11}\kappa_{12}}{{\zeta_1\zeta_2}}H(\zeta_{12}\tau)\right]+ \frac{G_1(S)- G_1(t)}{G_1^\prime(t)} X_1(t)\right\rbrace
		\end{multline}
		\begin{multline*}
		B_l(S,t)=\exp\left\lbrace-\mu\tau-\tau\frac{\theta_1\kappa_{11}}{\zeta_1}-\tau\frac{\theta_1\kappa_{11}}{\zeta_1}H(\zeta_1\tau)-\tau\frac{1}{2}\left(\frac{\kappa_{11}}{\zeta_1}\right)^2-\frac{1}{2}\tau\frac{\kappa_{11}\kappa_{12}}{{\zeta_1\zeta_2}}((H{\zeta_1+\zeta_2})\tau) \right. \\ \left. +\tau\frac{\kappa_{11}\kappa_{12}}{{\zeta_1\zeta_2}}H(\zeta_{12}\tau)+ \tau H(\zeta_1\tau)X_1(t)\right\rbrace\\
		=\exp\left\lbrace-\tau\left[\mu+\theta_1\frac{\kappa_{11}}{\zeta_1}-\frac{1}{2}\left(\frac{\kappa_{11}}{\zeta_1}\right)^2+H(\zeta_1\tau)\left[\theta_1\frac{\kappa_{11}}{\zeta_1}-\frac{\kappa_{11}\kappa_{12}}{{\zeta_1\zeta_2}}\right]\right. \right.\\ \left. \left.+\frac{1}{2}((H{\zeta_1+\zeta_2})\tau)\frac{\kappa_{11}\kappa_{12}}{{\zeta_1\zeta_2}}\tau- H(\zeta_1\tau)X_1(t)\right]\right\rbrace
		\end{multline*}
		As $\tau\rightarrow\infty$ 
		\begin{align*}
		\lim\limits_{\tau\rightarrow\infty}H(\zeta_1)&=0\\
		\lim\limits_{\tau\rightarrow\infty}H(\zeta_{12})&=0\\
		\lim\limits_{\tau\rightarrow\infty}((H{\zeta_1+\zeta_2})\tau)&=0
		\end{align*}
		\begin{equation*}
		\therefore B_l(S,t)=\exp\left\lbrace-\tau\left[R_l(\infty)+w(\tau)- H(\zeta_1\tau)X_1(t)\right]\right\rbrace \ \ \ t\in[0,S]
		\end{equation*}
	\end{proof}
	\section{Application of Kalman Filter Technique for Longevity Bond Price}
	The application of the Kalman filter to longevity bond price model is considered in this section. More especially, the application to one-factor longevity bond price model is considered as presented in the previous section. The Kalman filter equations are employed so that the unobservable variables and parameters of the longevity  bond price model can be estimated. To achieve this, the term structure of the models are represented in a linear state space form. This is to allow for measurement errors. 
	\subsection{State Space Formulation} 
	A state space model consists of a measurement equation relating the observed data and a Markovian transition equation that discribes the evolution of state vector over time. Here, state space model  of the longevity bond price model of one-factor Vasicek model is developed. After the state space form of the model has been written, the selection of $X_1$ is chosen by construction because it is generally not unique. That is, one the major aims of state space formulation is to set up $X_1$ such that it will contain the necessary information on the system at time $ t $.
	The basic criterion for a good state space representation is a minimal realisation.This is situation where the length of the state vector of the state space is minimised.
	\subsubsection{State Equation}
	The state equation is represented by the exact discrete-time distribution of the state variable gotten from equation \ref{dyna}.By solving the linear differential equation of the state variable, the exact discrete-time distribution is obtained.
	\begin{equation*}
	dX_1=-\zeta_1 X_1dt + c_1dW_1
	\end{equation*}
	\begin{align*}
	&X_1(t_n)=e^{-\zeta_1(t_n-t_{n-1})}X_1(t_{n-1})+\eta_1(t_n)\\
	&\left(\begin{array}{c}X_1^r(t_n)\\X_1^\lambda(t_n)\end{array}\right)=\left(\begin{array}{cc}
	e^{-\zeta_1^r(t_n-t_{n-1})}X_1^r(t_{n-1})\\e^{-\zeta_1^\lambda(t_n-t_{n-1})}X_1^\lambda(t_{n-1})
	\end{array}\right)+\left(\begin{array}{c}
	\eta_1^r(t_n)\\\eta_1^\lambda(t_n).
	\end{array}\right)
	\end{align*}
	\begin{equation*}
	\left(\begin{array}{c}X_1^r(t_n)\\X_1^\lambda(t_n)\end{array}\right)=\left(\begin{array}{cc}
	e^{-\zeta_1^r(t_n-t_{n-1})}&0\\0&e^{-\zeta_1^\lambda(t_n-t_n-1})
	\end{array}\right)\left(\begin{array}{c}
	X_1^r(t_{n-1})\\X_1^\lambda(t_{n-1})
	\end{array}\right)+\left(\begin{array}{c}
	\eta_1^r(t_n)\\\eta_1^\lambda(t_n).
	\end{array}\right)
	\end{equation*}
	Thus 
	\begin{equation*}
	X(t_n)=\gamma(\Phi)X(t_{n-1})+\eta(t_n)
	\end{equation*}
	\begin{equation}
	X_n=\gamma(\Phi)X_{n-1}+\eta_n
	\end{equation}
	\subsubsection{Measurement Equation}
	The measurement equation describes the relation between the  yields observed and mortality rates as measured with errors and the possible unobserved states.The measured errors are additive and normally distributed.
	Recall from equation \ref{const para} for yield of a zero coupon longevity bond with infinite maturity
	\begin{align}
	R_{l1}(t+\tau,t)&=R_{l1}(\infty)-w(\tau)-H(\zeta_1\tau)X_1(t)\\
	&=R_l(\infty)-w(\tau)-\left(\begin{array}{cc}
	H(\zeta_1^r\tau) &0\\
	0&H(\zeta_1^\lambda\tau)
	\end{array}\right)^\prime \left(\begin{array}{c}
	X_1^r(t)\\
	X_1^\lambda(t)
	\end{array}\right)\\
	&=A(\tau)-L(\tau)^\prime X_1(t)
	\end{align}
	where the transpose is denoted by the superscript,\\ and
	\begin{align*}
	A(\tau)=R_l(\infty)-w(\tau),\\
	L(\tau)=\left(\begin{array}{cc}
	H(\zeta_1^r\tau) &0\\
	0&H(\zeta_1^\lambda\tau)
	\end{array}\right)\\  and \ \ X_1(t)=\left(\begin{array}{c}
	X_1^r(t)\\
	X_1^\lambda(t)
	\end{array}\right)\\
	\end{align*}
	$A(\tau)$ is a scalar and are functions of the term to maturity $\tau$, then $L(\tau)$ is a $2\times2$\\and
	\begin{equation*}
	L(\tau)=\left(\begin{array}{cc}
	\frac{(1-e^{-\zeta_1^r\tau})}{\zeta_1^r\tau}&0\\
	0&\frac{(1-e^{-\zeta_1^\lambda\tau})}{\zeta_1^\lambda\tau}
	\end{array}\right)
	\end{equation*}
	where $\mu, \theta_1$ are parameters of the model.
	\begin{equation*}
	\begin{aligned}
	R_l(S,t)&=-\frac{log B_l(t_n+\tau_1,t_n)}{\tau_1}
	\\
	&=R_l(\infty)-w(\tau)-H(\zeta_1\tau_1)X_1(t_n)
	\\
	&= \left(\begin{array}{c}
	\mu^r\\ \mu^\lambda \end{array}\right)
	+ 
	\theta_1
	\left(\begin{array}{cc}
	\frac{\kappa_{11}}{\zeta_1^r}&0\\0& \frac{\kappa_{11}}{\zeta_1^\lambda}\end{array}\right) - \frac{1}{2}\left(\begin{array}{cc}
	\frac{\kappa_{11}}{\zeta_1^r}&0\\0& \frac{\lambda\kappa_{11}}{\zeta_1^\lambda}
	\end{array}\right)^2
	\\&
	-\left(\begin{array}{cc}H(\zeta_1^r)&0\\0&H(\zeta_1^\lambda)
	\end{array}\right)\left[\theta_1\left(\begin{array}{cc}
	\frac{\kappa_{11}}{\zeta_1^r}&0\\0&\frac{\kappa_{11}}{\zeta_1^\lambda}
	\end{array}\right)-\left(\begin{array}{cc}
	\frac{\kappa_{11}\kappa_{12}}{\zeta_1^r\zeta_2^r}&0\\0&\frac{\kappa_{11}\kappa_{12}}{\zeta_1^\lambda\zeta_2^\lambda}
	\end{array}\right)\right]
	\\&+\frac{1}{2}\left[\left(\begin{array}{cc}
	H(\zeta_1^r+\zeta_2^r)&0\\0&H(\zeta_1^\lambda+\zeta_2^\lambda)
	\end{array}\right)\tau\right]\left(\begin{array}{cc}\frac{\kappa_{11}\kappa_{12}}{\zeta_1^r\zeta_2^r}&0\\0&\frac{\kappa_{11}\kappa_{12}}{\zeta_1^\lambda\zeta_2^\lambda}\end{array}\right)
	\\&
	-\left(\begin{array}{cc}
	H(\zeta_1^r\tau) &0\\
	0&H(\zeta_1^\lambda\tau)
	\end{array}\right)\left(\begin{array}{c}
	X_1^r(t)\\
	X_1^\lambda(t)
	\end{array}\right)
	\end{aligned}
	\end{equation*} 
	The measurement equation is then given as:
	\begin{equation}
	R_{ln}=\gamma(\Phi)X(t_n)+d(\Phi)+\varepsilon_n, \ \ \varepsilon_n\sim N(0,H(\Phi))
	\end{equation}
	where $\Phi$ is made up the unknown parameters of the model  which includes the distribution of the measurement error. $d(\Phi)$ is an $N\times 1$ matrix whose row is given by $A(\tau_1;\Phi)$ and $\gamma(\Phi)$ is an $2\times 2$ matrix whose row is given as $-L(\tau_1;\Phi)^\prime$. \\ The  $d(\Phi)$ elements and that of $\gamma(\Phi)$ are then given by 
	\begin{equation*}
	d(\Phi)=A(\tau_1;\Phi)
	\end{equation*}
	and
	\begin{equation*}
	\gamma(\Phi)=-L(\tau_1;\Phi)^\prime
	\end{equation*}
	where 
	\begin{equation*}
	L(\tau_1;\Phi)=\left(\begin{array}{cc}
	H(\zeta_1^r\tau) &0\\
	0&H(\zeta_1^\lambda\tau)
	\end{array}\right) \qquad and \ \ A(\tau_1;\Phi)= R_l(\infty)-w(\tau)
	\end{equation*}
	Thus
	\begin{equation*}
	\gamma_1(\Phi)=\left(\begin{array}{cc}
	-H(\zeta_1^r\tau) &0\\
	0&-H(\zeta_1^\lambda\tau)
	\end{array}\right)
	\end{equation*}
	and 
	\begin{equation*}
	d(\Phi)=\left[R_l(\infty)-w(\tau)\right]
	\end{equation*}
	
	\subsection{Kalman Filter for Longevity Bond Price}
	The major purpose of the Kalman filter equations is to help in obtaining information about $X_n$ from the observed interest rates in the measurement equation. It also helps to estimate the unknown parameters of the models in  such a way that it includes all the relevant information on the equation at time $ t $.
	The basic criterion for a good state space representation is a minimal realisation. That is, a situation where the length of the state vector of the state space is minimised. The filter consists of two sets of equation
	\begin{enumerate}
		\item Prediction equations and
		\item Updating equations
	\end{enumerate}
	\subsubsection{Prediction Equations} 
	In the prediction equation, the information comes from the observed variables. it involves estimation of the unobserved state variables at a particular time based on the available information up to the time before that particular time. However, in this case, the unobserved variables are the observed interest rates which implies that the information comes from the observed interest rates.
	
	As the available information is given till time $t_{n-1}$,  the conditional expectation of the unknown vector space $X_n$ be denoted by $E_{n-1}(X_n)$. Also given the observed interest rates up to $t_{n-1}$ ,let the optimal estimator of $X_n$ be denoted by $X_{n|n-1}$. Hence  based on the observed interest rates up to time $t_{n-1}$, the optimal estimator of $X_n$ is the conditional expectation of $X_n$ given the available information up to time $t_{n-1}$. Thus we have 
	\begin{equation}
	\begin{aligned} \label{measurement}
	\hat{X}_{n|n-1}&= E_{n-1}(X_n) \\
	&= E_{n-1}(\gamma(\Phi)X_{n-1}+\eta_n),\\ 
	&=\gamma(\Phi)(\hat{X}_{n-1}).
	\end{aligned}
	\end{equation}
	where $ E_{n-1}(\hat{X}_{n-1})=X_{n|n-1} $ and $ E_{n-1}(\eta_n)=0$. From the observed longevity bond up to time $ t_{n-1} $, the prediction error is connected with estimating $ X_n $. Let $ \boldsymbol{B}_{n|n-1} $ represent the covariance marix of the estimator error.that is
	\begin{equation}
	\begin{aligned} \label {measurement1}
	\boldsymbol{B}_{n|n-1}&=E_{n-1}\left[ (X_n-\hat{X}_{n|n-1})(X_n-\hat{X}_{n|n-1})^\prime\right]\\
	&=\gamma(\Phi)E_{n-1}\left[ (X_n-\hat{X}_{n|n-1})(X_n-\hat{X}_{n|n-1})^\prime\right]\gamma^\prime(\Phi)+E_{n-1}\left[ \eta_n\eta_n^\prime\right] \\&=\gamma(\Phi)\boldsymbol{B}_{n-1}\gamma^\prime(\Phi)+V,
	\end{aligned}
	\end{equation}
	where $ E_{n-1}\left[ \eta_n\eta_n^\prime\right] = V$ is the covariance matrix of $\eta_n$.
	
	Thus equation \ref{measurement} and \ref{measurement1} are  the prediction eqautions which are used in the prediction step of the Kalman filter algorithem.
	\subsubsection{Updating Equations}
	The update procedure entails additional information on the longevity bond yield $R_{ln}$ at time $t_n$. This is to obtain more accurate and updated estimator of $X_n$. The new and present estimate of $X_n$ is called the filtered estimate. The moment new observed interest rates becomes available, $\hat{X}_{n|n-1}$ can then be updated to $\hat{X}_n$.\\ 
	While $\hat{X}_n$ denotes the optimal estimator of $X_n$  based on the available information from the observed interest rates up to time $t_n$ and $\boldsymbol{B}_n$ be the covariance matrix of the estimation error. Then the updating equations are given by 
	\begin{equation}
	\begin{aligned} \label{update 1}
	&\hat{X}_n=E_n(X_n)\\
	&\hat{X}_n= \hat{X}_{n|n-1}+\boldsymbol{B}_{n|n-1}Z^\prime F_n^{-1}v_n
	\end{aligned}
	\end{equation}
	and
	\begin{equation}
	\begin{aligned}
	\boldsymbol{B}_n&= E_n\left[ (X_n-\hat{X}_n)(X_n-\hat{X}_n)^\prime\right]\\
	&= \boldsymbol{B}_{n|n-1}-\boldsymbol{B}_{n|n-1}Z^\prime F_n^{-1}\boldsymbol{B}_{n|n-1} \\
	&= \left( \boldsymbol{B}_{n|n-1}^{-1}+ZH^{-1}Z\right) ^{-1},\label{justified}
	\end{aligned}
	\end{equation}
	where
	\begin{align}\label{residual}
	v_n &= R_n-(d+Z\hat{X}_{n|n-1})\\
	F_n&= H+Z\boldsymbol{B}_{n|n-1}Z^\prime
	\end{align}
	The log-likelihood function given by
	\begin{equation*}
	\begin{aligned}
	L&=-\frac{1}{2}\sum_{n=1}^{j}log\left| \boldsymbol{B}_n\right| -\frac{1}{2}\sum_{n=1}^{j}v_n\boldsymbol{B}_n^{-1}v_n\\
	L&=-\sum_{n=1}^{j}\left[\frac{1}{2}log\left| \boldsymbol{B}_n\right|+\frac{1}{2}v_n\boldsymbol{B}_n^{-1}v_n\right]
	\end{aligned} 
	\end{equation*}
	where $v_n$ is the correction term called measurement residual as given in equation \ref{residual}.
	The logarithm price is used to suit the Kalman filter’s linear structure.
	
	Using the updating equations, the optimal estimator at the next time step based on all the  available information up to the current time step can be obtained.
	\begin{equation}
	\hat{X}_{n+1|n}=\gamma(\Phi)\hat{X}_n+\eta_{n+1}\label{update ed 1}
	\end{equation}
	substituting equation\ref{measurement} and the value for $v_n$ into equation\ref{update ed 1} we have 
	\begin{equation}
	\begin{aligned}
	\hat{X}_{n+1|n}&=\gamma(\Phi)\left( \hat{X}_{n|n-1}+\boldsymbol{B}_{n|n-1}Z^\prime B_n^{-1}\left( R_{ln}-(d+Z\hat{X}_{n|n-1})\right)\right) +\eta_{n+1}\\
	&=\gamma(\Phi)\hat{X}_{n|n-1}+\gamma(\Phi)\boldsymbol{B}_{n|n-1}Z^\prime B_n^{-1}\left( R_{ln}-(d+Z\hat{X}_{n|n-1})\right)+\eta_{n+1}\\
	&=\gamma(\Phi)\hat{X}_{n|n-1}+\boldsymbol{K}_t\left( R_n-(d+Z\hat{X}_{n|n-1})\right)+\eta_{n+1} \\
	&=\left( \gamma(\Phi)-K_tZ\right)\hat{X}_{n|n-1}+K_tR_{ln}+\left(\eta_{n+1}-K_td\right)\label{Justifying}\end{aligned}
	\end{equation}
	from equation \ref{Justifying}, the Kalman gain matrix $ \boldsymbol{K}_t $ given by
	\begin{equation}
	\boldsymbol{K}_t=\gamma(\Phi)\boldsymbol{B}_{n|n-1}Z^\prime B_n^{-1}
	\end{equation}
	Eliminating $ F_n $ from $\boldsymbol{B}_n$ and then substituting the value for $\boldsymbol{B}_n$ in \ref{justified}
	\begin{equation*}
	\boldsymbol{B}_K=\boldsymbol{B}_{n|n-1}-\boldsymbol{B}_{n|n-1}Z^\prime(H+Z\boldsymbol{B}_{n|n-1}Z^\prime)^{-1}Z\boldsymbol{B}_{n|n-1}.
	\end{equation*}
	Recall, $ HH^{-1}=1 $ and $ (EG)^{-1}=G^{-1}E^{-1} $, we have
	\begin{equation*}
	\begin{aligned}
	(H+Z\boldsymbol{B}_{|n-1}Z^\prime)^{-1}=&\left( H(I+H^{-1}Z\boldsymbol{B}_{n|n-1}Z^\prime)\right) ^{-1}\\
	&=(I+H^{-1}Z\boldsymbol{B}_{n|n-1}Z^\prime)^{-1}H^{-1}
	\end{aligned}
	\end{equation*}
	Hence
	\begin{equation*}
	\boldsymbol{B}_n=\boldsymbol{B}_{n|n-1}-\boldsymbol{B}_{n|n-1}Z^\prime(I+H^{-1}Z\boldsymbol{B}_{n|n-1}Z^\prime)^{-1}H^{-1}Z\boldsymbol{B}_{n|n-1}
	\end{equation*}
	Note: The generalised inverse formular for the sum of invertible matrices as seen in \cite{Henderson and Searle (1981)} is given as
	\begin{equation}
	H=HU(I+BVHU)^{-1}BV H=(H^{-1}+UBV)^{-1}
	\end{equation}
	which gives
	\begin{equation} 
	\boldsymbol{B}_{n|n-1}-\boldsymbol{B}_{n|n-1}Z^\prime(I+H^{-1}Z\boldsymbol{B}_{n|n-1}Z^\prime)^{-1}H^{-1}Z\boldsymbol{B}_{n|n-1}=\left( \boldsymbol{B}_{n|n-1}^{-1}+Z^\prime H^{-1}Z\right)^{-1}
	\end{equation}
	\begin{equation}
	\boldsymbol{B}_{n|n-1}-\boldsymbol{B}_{n|n-1}Z^{-1}B_n^{-1}Z\boldsymbol{B}_{n|n-1}=\left( \boldsymbol{B}_{n|n-1}^{-1}+Z^\prime H^{-1}Z\right)^{-1}
	\end{equation}
	Thus the Kalman Filter equation for the longevity bond price is given as follows
	\begin{itemize}
		\item[1.] State Prediction
		\begin{equation*}
		\hat{X}_{n|n-1} = \gamma(\Phi)\hat{X}_{n-1|n-1}
		\end{equation*}
		\item[2.] Error Covariance Prediction
		\begin{equation*}
		\boldsymbol{B}_{n|n-1}=\gamma(\Phi)\boldsymbol{B}_{n-1|n-1}\gamma\prime(\Phi)+V
		\end{equation*}
		\item[3.] Kalman Gain
		\begin{equation*}
		\boldsymbol{K}_t=	\boldsymbol{B}_{n|n-1}ZF^{-1}_t
		\end{equation*}
		\item[4.] State Update
		\begin{equation*}
		\hat{X}_{n|n}= \hat{X}_{n|n-1}+\boldsymbol{K}_t(R_{ln}-(d+Z\hat{X}_{n|n-1}))
		\end{equation*}
		\item[5.] Error Covariance Update
		\begin{equation*}
		\boldsymbol{B}_{n|n}=\boldsymbol{B}_{n|n-1}-\boldsymbol{B}_{n|n-1}\boldsymbol{K}_tZ
		\end{equation*}
	\end{itemize}
	\section{Numerical Implementation}
	In this section, we we put into application all the preceding theoretical discussion of our estimation technique to our model.
	The exact discrete-time distribution of the state variable obtained from  equation \ref{dyna}\begin{equation*}
	dX_1=-\zeta_1 X_1dt + c_1dW_1
	\end{equation*}
	By solving the linear differential equation of the state variable, the exact discrete-time distribution is obtained, 
	\begin{equation*}
	X_n=\gamma(\Phi)X_{n-1}+\eta_n
	\end{equation*}
	The measurement equation is then given as:
	\begin{equation*}
	R_{ln}=\gamma(\Phi)X(t_n)+d(\Phi)+\varepsilon_n, \ \ \varepsilon_n\sim N(0,H(\Phi))
	\end{equation*}
	where $\Phi$ is made up the unknown parameters of the model  which includes the distribution of the measurement error. $d(\Phi)$ is an $N\times 1$ matrix whose row is given by $A(\tau_1;\Phi)$ and $\gamma(\Phi)$ is an $2\times 2$ matrix whose row is given as $-L(\tau_1;\Phi)^\prime$. \\ The  $d(\Phi)$ elements and that of $X(\Phi)$ are then given by 
	\begin{equation*}
	d(\Phi)=A(\tau_1;\Phi)
	\end{equation*}
	and
	\begin{equation*}
	\gamma(\Phi)=-L(\tau_1;\Phi)^\prime
	\end{equation*}
	where 
	\begin{equation*}
	L(\tau_1;\Phi)=\left(\begin{array}{cc}
	H(\zeta_1^r\tau) &0\\
	0&H(\zeta_1^\lambda\tau)
	\end{array}\right) \qquad and \ \ A(\tau_1;\Phi)= R(\infty)-w(\tau)
	\end{equation*}
	Thus
	\begin{equation*}
	\gamma_1(\Phi)=\left(\begin{array}{cc}
	-H(\zeta_1^r\tau) &0\\
	0&-H(\zeta_1^\lambda\tau)
	\end{array}\right)
	\end{equation*}
	and 
	\begin{equation*}
	d(\Phi)=\left[R_l(\infty)-w(\tau)\right]
	\end{equation*}
	The predicted equation is obtained by given the available information up to time $t_{n-1}$. As the available information is given till time $t_{n-1}$,  the conditional expectation of the unknown vector space $X_n$ be denoted by $E_{n-1}(X_n)$. Also given the observed interest rates up to $t_{n-1}$ ,let the optimal estimator of $X_n$ be denoted by $X_{n|n-1}$. Hence  based on the observed interest rates up to time $t_{n-1}$, the optimal estimator of $X_n$ is the conditional expectation of $X_n$ given the available information up to time $t_{k-1}$. That is 
	\begin{equation*}
	\begin{aligned}
	\hat{X}_{n|n-1}&= E_{n-1}(X_n) \\
	&=\gamma(\Phi)(\hat{X}_{n-1}),\label{predict}
	\end{aligned}
	\end{equation*}
	where $ E_{n-1}(\hat{X}_{n-1})=X_{n|n-1} $ and $ E_{n-1}(\eta_n)=0. $. Considering the error involved in estimating $ X_n $ during the prediction step will be the next thing. From the observed longevity bond up to time $ t_{n-1} $, the prediction error is connected with estimating $ X_n $. Thus the estimator error is given as
	\begin{equation*}
	\begin{aligned}
	\boldsymbol{B}_{n|n-1}&=E_{n-1}\left[ (X_n-\hat{X}_{n|n-1})(X_n-\hat{X}_{n|n-1})^\prime\right]\\
	&=\gamma(\Phi)\boldsymbol{B}_{n-1}\gamma^\prime(\Phi)+V,\label{error}
	\end{aligned}
	\end{equation*}
	where $ E_{n-1}\left[ \eta_n\eta_n^\prime\right] = V$ is the covariance matrix of $\eta_n$ and $\boldsymbol{B}_{n|n-1}$ is the covariance matrix of the estimated error.
	
	The updating equation is given as
	\begin{equation*}
	\begin{aligned}
	&\hat{X}_n=E_n(X_n)\\
	&\hat{X}_n= \hat{X}_{n|n-1}+\boldsymbol{B}_{n|n-1}Z^\prime F_n^{-1}v_n\label{update}
	\end{aligned}
	\end{equation*}
	and
	\begin{equation*}
	\begin{aligned}
	\boldsymbol{B}_n&= E_n\left[ (X_n-\hat{X}_n)(X_n-\hat{X}_n)^\prime\right]\\
	&= \boldsymbol{B}_{n|n-1}-\boldsymbol{B}_{n|n-1}Z^\prime F_n^{-1}\boldsymbol{B}_{n|n-1} \\
	&= \left( \boldsymbol{B}_{n|n-1}^{-1}+ZH^{-1}Z\right) ^{-1},\label{justify}
	\end{aligned}
	\end{equation*}
	where
	\begin{align*}
	v_n &= R_{ln}-(d+Z\hat{X}_{n|n-1})\\
	F_n&= H+Z\boldsymbol{B}_{n|n-1}Z^\prime
	\end{align*}
	Thus the log-likelihood function given by
	\begin{equation*}
	L=-\sum_{n=1}^{j}\left[\frac{1}{2}log\left| \boldsymbol{B}_n\right|+\frac{1}{2}v_n\boldsymbol{B}_n^{-1}v_n\right] 
	\end{equation*}
	where $v_n$ is the correction term called measurement residual as given in equation \ref{residual}.
	The logarithm price is used to suit the Kalman filter’s linear structure.
	
	The next optimal estimator will be based on all the available information up to the current time step which is gotten using the updating equation
	\begin{equation*}
	\hat{X}_{n+1|n}=\gamma(\Phi)\hat{X}_n+\eta_{n+1}\label{update equ}
	\end{equation*}
	\begin{equation*}
	\begin{aligned}
	\hat{X}_{n+1|n}&=\gamma(\Phi)\left( \hat{X}_{n|n-1}+\boldsymbol{B}_{n|n-1}Z^\prime B_n^{-1}\left( R_{ln}-(d+Z\hat{X}_{n|n-1})\right)\right) +\eta_{n+1}\\
	&=\left( \gamma(\Phi)-\boldsymbol{K}_tZ\right)\hat{X}_{n|n-1}+\boldsymbol{K}_tR_n+\left(\eta_{n+1}-\boldsymbol{K}_td\right)\label{Justification}\end{aligned}
	\end{equation*}
	and
	\begin{equation*}
	\boldsymbol{K}_t=\gamma(\Phi)\boldsymbol{B}_{n|n-1}Z^\prime B_n^{-1}
	\end{equation*}
	\pagebreak
	\begin{table}
		\begin{center}
			\begin{tabular}{ |p{3cm}|p{3cm}|p{3cm}| }
				
				\hline
				\multicolumn{3}{|c|}{Results} \\
				\hline
				Parameters & True & Estimated    \\
				\hline\hline
				$\mu$ & 0.1& 0.97 \\
				\hline
				$\zeta$ &1 &2.06 \\
				\hline
				$c_1$ & 3& 3.1 \\
				\hline
			\end{tabular}
		\end{center}
		\caption{Result of the graph below}
	\end{table}
	Measurement Noise Covariance $= R_l = 0.0043$
	
	The kalman filter tends to identify very closely the mean-reversion, long-term mean and volatility parameters and also offers a good approximation of the true parameters from a relatively small sampling space using the generated observed prices.
	\begin{figure}[!h]
		\begin{center}
			\includegraphics[scale=0.55]{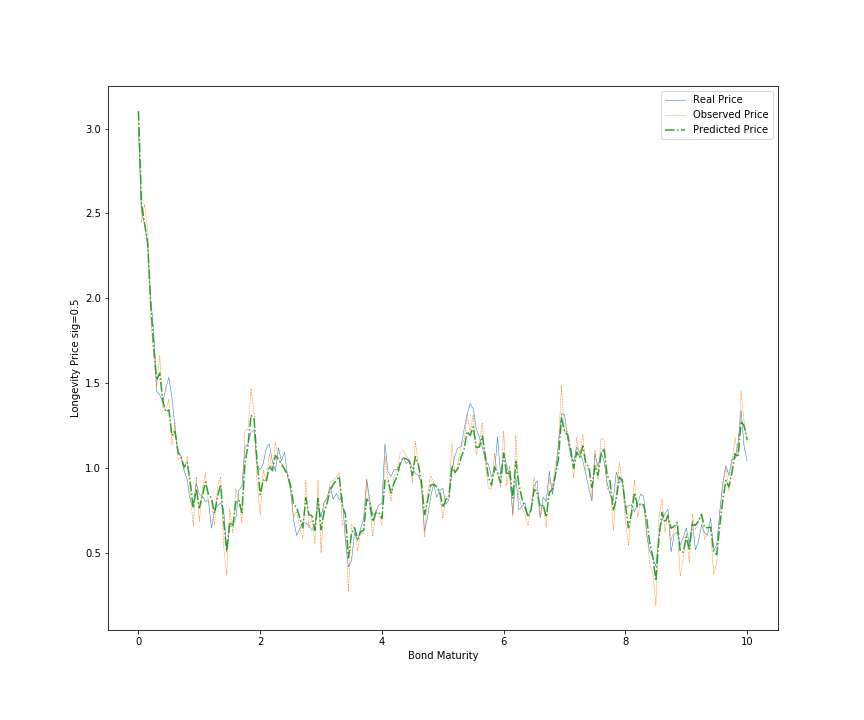}
			\caption{Simulation of longevity Bond Price}
			\label{figure 3.1:Vasicek Longevity Bond Simulation }
		\end{center}
	\end{figure}
	\pagebreak
	\subsection{Model Simulation Results}
	The results of the actual longevity bond market data is presensented in this section. The U.S. Treasury Bonds data as used in this work was downloaded from the website \textcolor{blue}{Yahoo finance}. The raw data includes the Date, Open, High, Low, Close, Adjusted (Adj) Close and Volume. The data is a daily quotes ranging from $31st$ of December 1992 to the $8th$ of November 2017 with 10 years maturity.
	\pagebreak
	\begin{table}
		\begin{tabular}{ |p{1cm}|p{2cm}|p{1.5cm}|p{1.5cm}|p{1.5cm}|p{1.5cm}|p{2cm}|p{1.5cm}| }
			\hline
			\multicolumn{8}{|c|}{U.S Treasury Bond} \\
			\hline
			&Date &	Open &	High &	Low &	Close &	Adj Close &	Volume\\
			\hline\hline
			0 &	1992-12-31& 	6.70 &	6.70 &	6.70 &	6.70 &	6.70 &	0.0\\
			\hline
			1 &	1993-01-04 &	6.60 &	6.60 &	6.60 &	6.60 &	6.60 &	0.0\\
			\hline
			2 &	1993-01-05 &	6.61 &	6.61 &	6.61 &	6.61 &	6.61 &	0.0\\
			\hline
			3 &	1993-01-06 &	6.63 &	6.63 &	6.63 &	6.63 &	6.63 &	0.0\\
			\hline
			4 &	1993-01-07 &	6.76 &	6.76 &	6.76 &	6.76 &	6.76 &	0.0\\
			\hline
			5 &	1993-01-08 &	6.75 &	6.75 &	6.75 &	6.75 &	6.75 &	0.0\\
			\hline
			6 &	1993-01-11 &	6.71 &	6.71 &	6.71 &	6.71 &	6.71 &	0.0\\
			\hline
			7 &	1993-01-12 &	6.72 &	6.72 &	6.72 &	6.72 &	6.72& 	0.0\\
			\hline
			8 &	1993-01-13&	6.71 &	6.71& 	6.71 &	6.71 &	6.71 &	0.0\\
			\hline
			9 &	1993-01-14 &	6.65 &	6.65 &	6.65 &	6.65 &	6.65& 	0.0\\
			\hline
			10 &	1993-01-15 &	6.60& 	6.60 &	6.60 &	6.60 &	6.60 &	0.0\\
			\hline
			11& 	1993-01-18& 	NaN &	NaN &	NaN &	NaN &	NaN &	NaN\\
			\hline
			12 &	1993-01-19 &	6.59& 	6.59 &	6.59& 	6.59 &	6.59 &	0.0\\
			\hline
			13 &	1993-01-20 &	6.61& 	6.61 &	6.61 &	6.61 &	6.61& 	0.0\\
			\hline
			14 &	1993-01-21 &	6.60 &	6.60 &	6.60 &	6.60 &	6.60 &	0.0\\
			\hline
			15 &	1993-01-22 &	6.57 &	6.57 &	6.57 &	6.57 &	6.57& 	0.0\\
			\hline
			16& 	1993-01-25 &	6.48 &	6.48 &	6.48 &	6.48 &	6.48 &	0.0\\
			\hline
			17 &	1993-01-26& 	6.50 &	6.50 &	6.50 &	6.50 &	6.50 &	0.0\\
			\hline
			18 &	1993-01-27& 	6.48 &	6.48 &	6.48 &	6.48 &	6.48 &	0.0\\
			\hline
			19& 	1993-01-28 &	6.44 &	6.44& 	6.44 &	6.44 &	6.44 &	0.0\\
			\hline
		\end{tabular}
		\caption{U.S Treasury Bond Maturity}
	\end{table}
	
	Adj Close price is used in adjusting the closing price of a stock in order to correctly reflect the value of that
	stock after accounting for any corporate shares. The estimated parameter set are then given in the table 3
	below.
	\pagebreak
	\begin{table}
		\begin{center}
			\begin{tabular}{ |p{3cm}|p{3cm}|p{3cm}| }
				\hline
				\multicolumn{3}{|c|}{Estimated Parameter Set} \\
				\hline
				Parameters & True & Estimated    \\
				\hline\hline
				$\mu$ & 3& 3.19 \\
				\hline
				$\zeta$ &0.5 &0.4 \\
				\hline
				$c_1$ & 1& 1.4 \\
				\hline
			\end{tabular}
		\end{center}
		\caption{Estimated Parameter Set}
	\end{table}
	Measurement Noise Covariance $= H = 0.0102$
	\begin{figure}[!h]
		\begin{center}
			\includegraphics[scale=0.55]{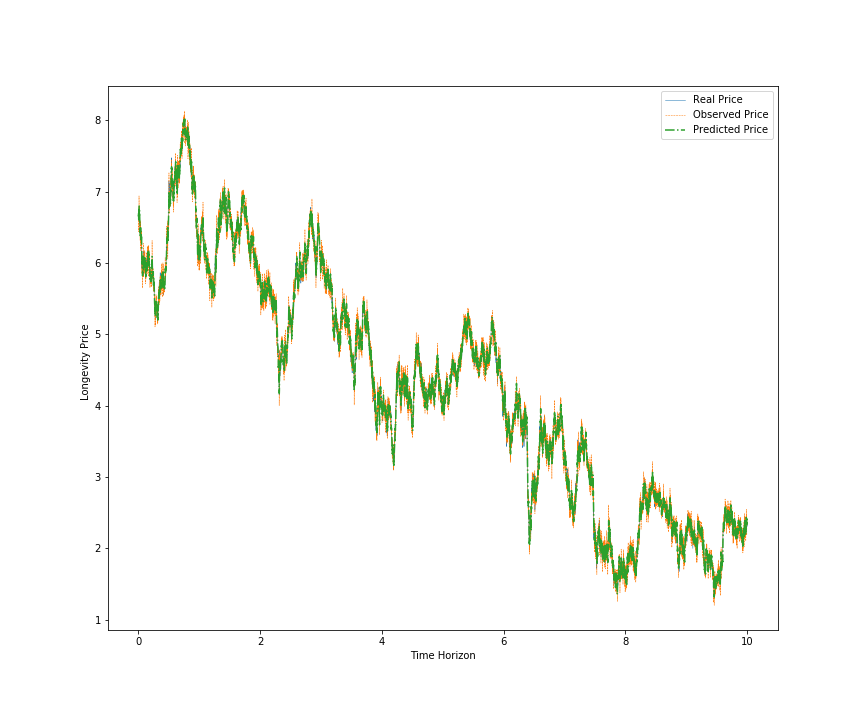}
			\caption{Simulation of longevity Bond Price}
			\label{figure 3.1:Vasicek Longevity Bond Simulation }
		\end{center}
	\end{figure}
	\newpage
	Figure 1 and Figure 2, shows that a mean reversion cycle matches the real price (solid
	line).An increase in the  measurement noise which is usually distributed to the real price gives the observed
	value system (thin line). The predicted price is calculated using the Kalman filter (dot-dashed line)
	which is real close to the observed price. The Kalman filter with Gaussian distributions is quite robust and adaptable to univariating and multivariating economic system variables. 
	
	The application of a Kalman filter in this instance to match
	price information basically replicates a linear fitting routine.Also the application of Kalman filter to estimate the parameter set using maximum likehood to optimize our
	parameter set proves to be almost realistic. Although there are some errors in our estimate parameters, however, our predicted price looks more likely like the real prices which shows that Kalman filter is a good method to estimate unobservable parameter set. To observe the action of mean reversion, we simulated the stochastic behaviour of bond price with a mean reversion speed equal to $0.5 (\zeta = 0.5)$ to observe the weak mean reversion.
	
	\section{Conclusion}
	In this paper a one-factor Vasicek model was estimated for the U.S. Treasury term structure from  December, 1992 to January, 1993. We modeled a linear state-space continuous time term structure model. Firstly, we specified a time series process for the instantaneous spot interest rate where the zero coupon longevity bond price formula is a function of the model's parameters and the unobserved instantaneous spot rate. The parameters  of the model are the long-run mean rate, the mean reversion towards the long-run mean, diffusion coefficients of the short-term interest rate and the market price of risk. The model is estimated using a maximum likelihood approach which is based on the Kalman filter. The Kalman filter recursive algorithm uses observable data on bonds to extract values for the unobserved state variables which combines the cross section and time series information in the term structure.
	
	The observed yields on zero-coupon longevity bonds is modelled as a linear function of one-factor state variable. Result of the empirical analysis is based on monthly observations of U.S. Treasury zero coupon bonds from December, 1992 to January, 1993. Ten maturities have been chosen which spans across the yield curve from 1 year to 20 years. The yield curve is expected to include influences on the short, medium and infinite horizon of the term structure. The model  parameters  and their standard errors are estimated.
	
	The empirical results of the analysis show that the unobserved instantaneous interest rate tends to exhibit a mean reverting behaviour in the U.S. term structure. The empirical evidence also indicates that to model adequately historical mortality trends at different ages, both the survival rate and the yield data are necessary for an optimal satisfactory for the empirical fit over the zero coupon longevity bond term structure. The result from the model dynamics allowed us to simulate the survival rates, which enabled us to model longevity risk. The affine form of longevity bond price models is able to be used not only for pricing and hedging longevity risk for pension funds and  insurers but also for evaluating capital requirements for risk management. 
	
	\section*{Acknowlegment}
	My sincre gratitude to University of South Africa for giving me Bursary funding thereby making the financial burden lessened for this research. Also to AIMS South Africa for allowing me use their research facilities during the period of this research.\newpage
	

\end{document}